\documentclass[manuscript,screen]{acmart}
\pushQED{\hfill$\natural$}%
\usepackage{fancyhdr, hyperref}
\usepackage{algorithmic}
\usepackage{algorithm}
\usepackage{subcaption}
\usepackage[T1]{fontenc}
\usepackage[normalem]{ulem}

\usepackage{enumitem}
\usepackage{etoolbox}
\usepackage{breqn}
\usepackage{lineno}
\usepackage{maple}
\usepackage{xfrac}
\usepackage{xcolor}
\definecolor{mygreen}{rgb}{0,0.6,0}
\definecolor{mygray}{rgb}{0.5,0.5,0.5}
\definecolor{mymauve}{rgb}{0.58,0,0.82}
\definecolor{altblue}{rgb}{0.0,0.6,1.0}
\definecolor{lstbg}{cmyk}{0.05, 0.01, 0, 0}
\definecolor{morebluish}{cmyk}{0.06,0.04,0,0}
\usepackage{listings}
\input{listings-maple-definition.sty}
\newcommand{\diag}[1]{\ensuremath{\mathrm{diag}{#1}}}

\newtheorem{lemma}{Lemma}
\newcounter{remarkcounter}[section]
\newenvironment{remark}[1][]{\refstepcounter{remarkcounter}\par\medskip
   \noindent \textbf{Remark~\theremarkcounter. #1} \rmfamily}{\medskip}

\usepackage{xfrac}
\usepackage{bm}

\newcommand{\mat}[1]{\ensuremath{\boldsymbol{#1}}}
\newcommand{\A}{\ensuremath{\mat{A}}}
\newcommand{\D}{\ensuremath{\mat{D}}}
\newcommand{\E}{\ensuremath{\mat{E}}}

\renewcommand{\H}{\ensuremath{\mat{H}}}

\renewcommand{\L}{\ensuremath{\mat{L}}}
\newcommand{\M}{\ensuremath{\mat{M}}}

\renewcommand{\P}{\ensuremath{\mat{P}}}

\renewcommand{\S}{\ensuremath{\mat{S}}}
\newcommand{\T}{\ensuremath{\mat{T}}}
\renewcommand{\emph}[1]{\textsl{#1}}
\copyrightyear{2022}
\acmYear{2022}
\setcopyright{acmlicensed}\acmConference[]{}{}
\acmBooktitle{}
\acmPrice{}
\acmDOI{}
\acmISBN{not-an-isbn-number}

\begin{document}
\fancyhead{}

\title[Bohemian Geometry]{Bohemian Matrix Geometry}

\author{Robert M.~Corless}
\email{rcorless@uwaterloo.ca}
\orcid{0000-0003-0515-1572}
% \author{Mark Giesbrecht}
% \email{mwg@uwaterloo.ca}

\affiliation{%
  \institution{Cheriton School of Computer Science\\ University of Waterloo}
  \country{Canada}
  }
\author{George Labahn}
\email{glabahn@uwaterloo.ca}
\affiliation{%
  \institution{Cheriton School of Computer Science\\ University of Waterloo}
  \country{Canada}
}

\author{Dan Piponi}
\email{dan.piponi@epicgames.com}
\orcid{0000-0002-6599-2835}
\affiliation{\institution{Epic Games}\country{USA}}

\author{Leili Rafiee Sevyeri}
\orcid{0000-0002-7819-9764}
\email{leili.rafiee.sevyeri@uwaterloo.ca}
\affiliation{%
  \institution{Cheriton School of Computer Science\\ University of Waterloo}
  \country{Canada}
}

%\renewcommand{\shortauthors}{Corless, et al.}

\begin{abstract}
A Bohemian matrix family is a set of matrices all of whose entries are drawn from a fixed, usually discrete and hence bounded, subset of a field of characteristic zero.  Originally these were integers---hence the name, from the acronym BOunded HEight Matrix of Integers (BOHEMI)---but other kinds of entries are also interesting.  Some kinds of questions about Bohemian matrices can be answered by numerical computation, but sometimes exact computation is better.  In this paper we explore some Bohemian families (symmetric, upper Hessenberg, or Toeplitz) computationally, and answer some  open questions posed about the distributions of eigenvalue densities.
\end{abstract}



\maketitle
\setlength{\parindent}{5mm} % Bah.
\section{Introduction}

 \begin{figure}
     \centering
     \includegraphics[width=0.9\textwidth]{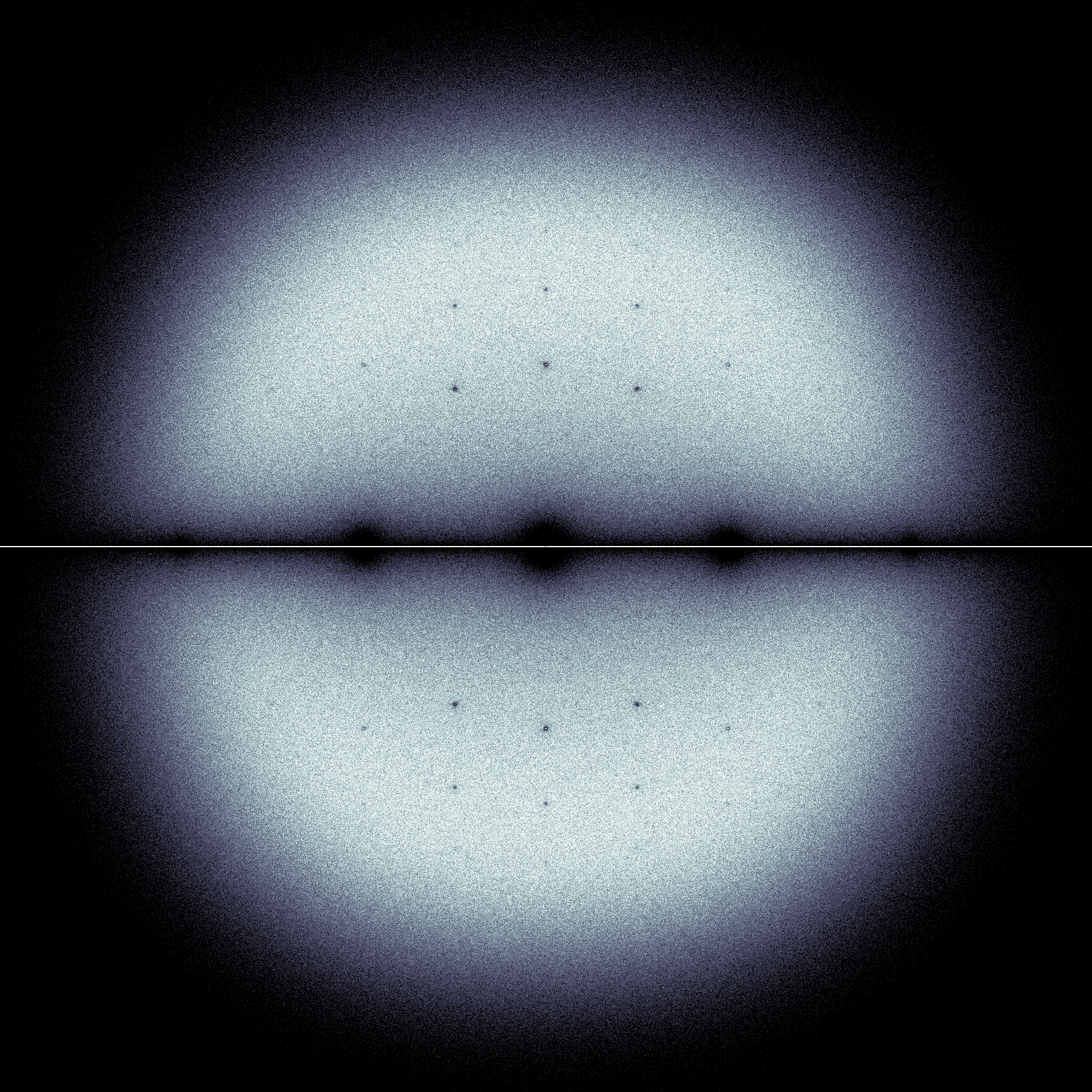}
    \caption{Eigenvalues of $5\cdot 10^6$ dimension $m=8$ matrices with entries chosen uniformly at random among the population $P=\{-1,0,1\}$.  We see that the uniform distribution as $m\to\infty$ result of~\cite{tao2006random} is evident already.  As $m$ increases, the ``holes'' close up, and the relative percentage of the real eigenvalues (which extend past the $\sqrt{m}$ disk radius) becomes negligible. }
     \label{fig:uniform}
     \Description{A fuzzy-edged light-grey disk, bisected with a dark line with dark bumps at $0$ and $\pm 1$.  Reminiscent of the moon occulted by a dark shadow.  Smaller holes are placed symmetrically on on near-circular arcs. }
 \end{figure}
A \textsl{Bohemian matrix family} is a set of matrices all of whose entries are drawn from a fixed finite population, usually integers, algebraic integers, or Gaussian integers.
The name ``Bohemian'' was invented in 2015 at the Fields Institute Thematic Year in Symbolic Computation; the mnemonic is useful because it highlights searching for commonality among features of matrices with discrete populations.
Our original interest was for \textsl{software testing}, and as a testing ground for optimization over (in search of improved computational bounds for certain quantities, such as the growth factor in Gaussian elimination with complete pivoting, or the departure from normality).  Bohemian matrices are a \textsl{specialization} in the sense of P\'olya, and have led now to several workshops, at \href{https://nhigham.com/2018/10/15/bohemian-matrices-in-manchester/}{Manchester in 2018}, at ICIAM in 2019, and at SIAM in 2021. There have been several publications since, including~\cite{Chan2020}, \cite{fasi2020determinants}, \cite{DudaFonsecaXuYe+2021+505+514}, \cite{Sendra2022}, and the very interesting~\cite{bogoya2022upper} which explores a connection to the asymptotic spectral theory of Toeplitz matrices~\cite{Schmidt1960}, which is very much alive today: see e.g.~\cite{bottcher2005spectral,bottcher2012introduction,barrera2018eigenvalues,garoni2017generalized} and~\cite{Bottcher2021}.

The study of matrices with rational integer entries is very old, and the literature too vast to survey coherently here.  We instead point to the early survey by Olga Taussky--Todd~\cite{Taussky1960} as an entry point. We are also going to be working with Gaussian integer and algebraic integer entries; see for instance~\cite{butson1962generalized} for important work on generalized Hadamard matrices where the entries are roots of unity.

The study of \textsl{random} matrices where the entries are drawn from discrete distributions is also very advanced; see~\cite{tao2006random,tao2017random} for instance. Those papers established that dense square matrices of dimension $m$ whose entries are drawn from a discrete population, say $-1$, $0$, and $1$, have eigenvalues that are asymptotically \emph{uniformly distributed} on a disk of radius $\sqrt{m}$.  See Figure~\ref{fig:uniform}.

However, if the matrices are \emph{structured}, other pictures arise, and little is known about the asymptotic distributions of their eigenvalues.  For instance, in~\cite{Corless:2021:WhatCan} we find \emph{skew-symmetric tridiagonal matrices} with unexpectedly \emph{square} distributions, or even diamond-shaped, as in Figure~\ref{fig:diamond}.
We will explain some of these shapes in this paper, and prove that they will \emph{not} fill out to disk shapes as $m \to \infty$.  We will also explain some other interesting features that arise in certain Bohemian families, including upper Hessenberg Toeplitz Bohemians.

\begin{figure}
    \centering
    \includegraphics[width=0.9\textwidth]{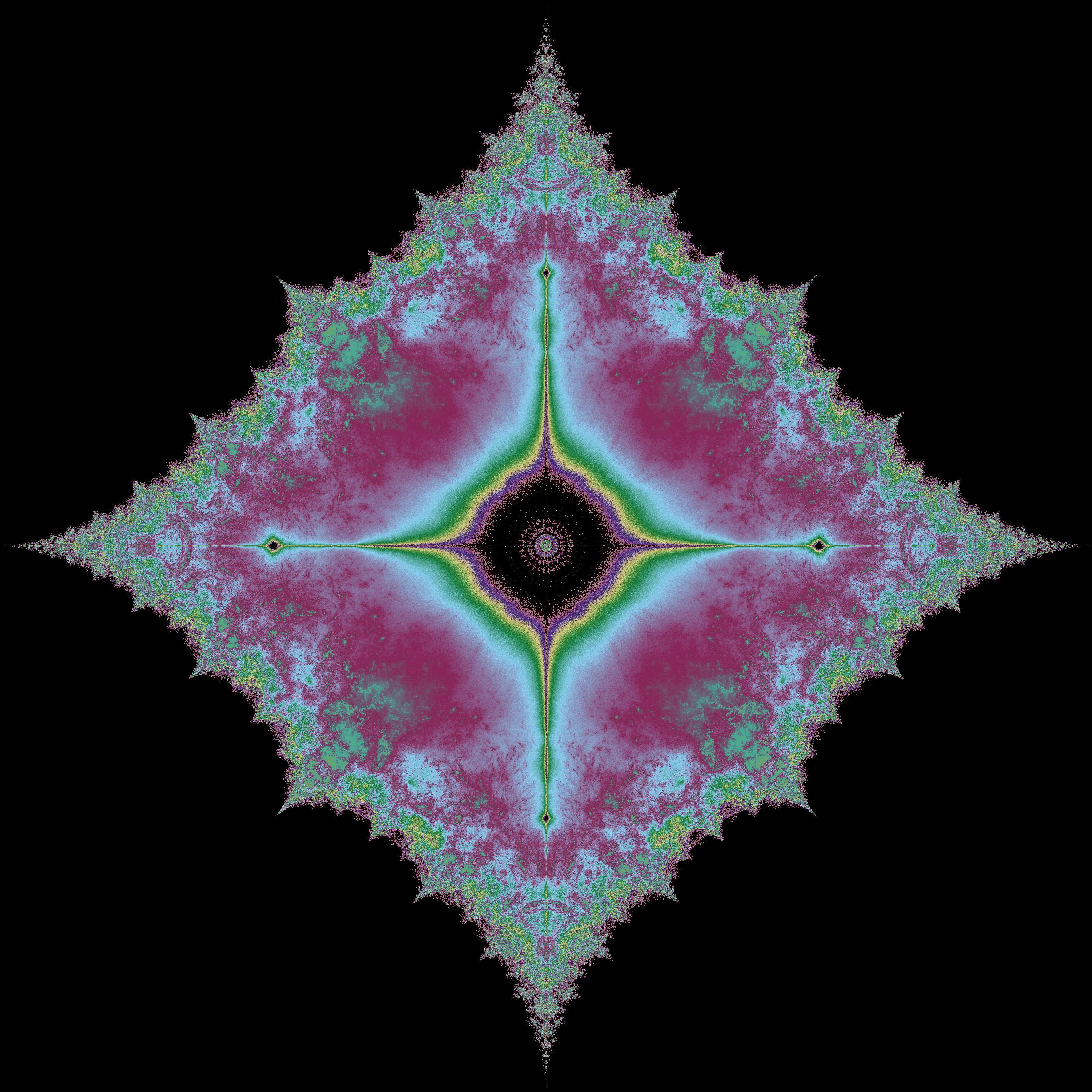}
    \caption{Eigenvalues of \emph{all} dimension $m=31$ skew-symmetric tridiagonal matrices with entries drawn from $P = \{-1,i,1,-i\}$, the fourth roots of unity. Picture courtesy Aaron Asner. We explain the unexpected eigenvalue geometry in this paper.}
\Description{A pinkish fractal diamond shape with a hole in the center outlined in yellow and red with a green and blue cross surrounding the hole.  The fractal edges on the outside are mostly green. The fractal edges are reminiscent of a saw blade. The sharpest peaks are at the four main corners $(\pm\sqrt{2},0)$, $(0,\pm\sqrt{2})$.}
    \label{fig:diamond}
\end{figure}

There is a significant connection to number theoretic works, as well. Kurt Mahler~\cite{mahler1963two} was interested in the distribution of zeros of polynomials with given \textsl{length} (the one-norm of the vector of coefficients) and \textsl{height} (the infinity-norm of the vector of coefficients). This is connected with the Littlewood conjecture for polynomials~\cite{littlewood1968} (How large on the unit circle must a polynomial with $-1$,$1$ coefficients be?).  The numerical visualization of zeros of polynomials with coefficients $0$ or $1$  was apparently first done in~\cite{Odlyzko1992}, who proved that the limiting set was connected; later work explained the ``holes"~\cite{borwein1997polynomials} and visualizations by Peter Borwein and Loki J\"orgenson made several other questions clearer~\cite{borwein2001visible}.  See also the web pages of \href{https://math.ucr.edu/home/baez/roots/}{John Carlos Baez} and of \href{https://jdc.math.uwo.ca/roots/}{Dan Christensen}. Their article, \textsl{The Beauty of Roots}, published on Baez' website at the previous link, explains quite a few of the visible structures. Then we will see a connection to Kate Stange's work on Schmidt Tessellations,  \href{https://math.katestange.net/illustration/schmidt-arrangements/}{https://math.katestange.net/illustration/schmidt-arrangements/}.  See also~\cite{harriss2020algebraic} and~\cite{dorfsman2021searching} who connect Galois theory and visualization of roots of polynomials (and therefore, although they do not point this out, of eigenvalues of Bohemian matrices).

It is only a small jump from polynomials of bounded height to \textsl{matrices} of bounded height; but the questions become (to our minds) even more interesting.

\section{Organization of the paper}
In section~\ref{sec:questions} we mention a few research questions about Bohemian matrices that may be interesting to the computer algebra community; in~\ref{sec:structures} we discuss specific matrix structures and give our main theorems, which describe and explain the constraints that these matrix structures place on the spectra. In particular, we give a new (and quite surprising) theorem about eigenvalues of upper Hessenberg Toeplitz matrices which gives an essentially complete explanation of the ``fractal'' edges seen in some of the figures. This is based on the well-known asymptotic spectral theory of Toeplitz matrices, but extended to the case where we have an uncountable number of such matrices in the limit as the dimension goes to infinity.  This new theorem extends a result of Schmidt and Spitzer, which is concerned with Toeplitz matrices whose ``symbol'' is a Laurent polynomial and with certain semi-algebraic curves that arise from that Laurent polynomial, to matrices whose symbol is a Laurent \emph{series}.

Together these theorems explain the appearance of some of these figures.
We also explain some of the ``algebraic number starscape" appearance~\cite{harriss2020algebraic} and connect to Schmidt tesselations~\cite{Stange2017} by making an \textsl{approximate} computation of eigenvalues, and displaying the results in~Figure \ref{fig:Rayleigh}.

\section{Some questions of interest\label{sec:questions}} %we are interested in}
Every polynomial written in the monomial basis can be embedded as a Frobenius companion matrix into a matrix of the same \textsl{height} (the height of a matrix $\mathbf{A}$, as opposed to a polynomial, is the infinity norm of the matrix reshaped into a vector). Therefore every question about roots of polynomials of bounded height translates directly into a question about eigenvalues of Bohemian matrices.  It will become clear as we go that this is one-way, that is,  there are questions of Bohemian matrices that do not translate into questions about bounded height polynomials.

One question is ``which matrices in the family have the largest \textsl{characteristic height}?''  The characteristic height is the height of the characteristic polynomial; as previously noted, the characteristic height might be exponentially larger than the matrix height.  This is so for certain upper Hessenberg Toeplitz matrices~\cite{Chan2020}, where a lower bound containing a Fibonacci number is given for the maximum characteristic height in the family studied in that paper.

% We remark that for populations of uniform absolute value (such as roots of unity) there is a connection between \textsl{condition number} and determinant: In~\cite{guggenheimer1995simple} we find the following theorem:
% \begin{theorem}
% For a row-equilibrated square matrix $\mathbf{M}$ of dimension $m$, the $2$-norm condition number $\kappa(\mathbf{M}) = \|\M\|_2\| \M^{-1}\|_2 < 2/|\det\mathbf{M}|$ and the constant $2$ is the best possible.
% \end{theorem}
% By \textsl{row-equilibrated} is meant that each row has been scaled (by multiplication by a diagonal matrix if necessary) so that the vector $2$--norm of each row is $1$.  If, for instance, each entry of the matrix has absolute value $B$, then the $2$-norm of each row of the matrix is $\sqrt{m}B$. Put $\mathbf{M} = \mathbf{A}/(B\sqrt{m})$ and so $\kappa({\M}) = \Vert\mathbf{M}\Vert\Vert\mathbf{M}^{-1}\Vert = \Vert\mathbf{A}\Vert\Vert\mathbf{A}^{-1}\Vert < 2/|\det\mathbf{M}| = 2(B\sqrt{m})^m/|\det\mathbf{A}|$.

% Thus, in contrast to the usual case where the size of the determinant is quite decoupled from the conditioning, we see that for some classes of Bohemian matrices the size of the determinant actually tells us something. [The connection can be weak; this upper bound can be quite pessimistic.  Nonetheless it is a useful bound.]

The characteristic height is connected to the numerical conditioning of the characteristic polynomial; indeed one may take the Lebesgue constant for the polynomial~\cite[ch.~8]{CorlessFillion2013} on the interval $-1 \le x \le 1$ to be the characteristic height.

% One can ask about the distribution of condition numbers in a family: how likely is a matrix picked ``at random'' from the family to be ill-conditioned?  One can ask about the distribution of \textsl{eigenvalue} condition numbers~\cite{armentano2019polynomial}.  See also~\cite{Ratnarajah2004}. How many matrices in the family are singular?  How many have multiple eigenvalues?  How many different eigenvalues are there?  How many matrices have nontrivial Jordan structure?  How many matrices have nontrivial integer Smith form?  How many different characteristic polynomials are there in the family? How many different minimal polynomials are there? How many non-normal matrices are there?  What is the typical departure from normality? How many orthogonal matrices are there? How many pairs of commuting matrices are there?  What is the distribution of eigenvalues?  We will typically be interested in complex eigenvalues, but the distribution of real eigenvalues is a classical topic in random matrix theory.

% There are an infinite number of further questions that can be asked.  We can then ask them again about block matrices; about matrix polynomials; and about commuting matrix families for multivariate polynomial equations.

Typically, one asks these questions in an asymptotic form; one wants answers valid in the limit as the dimension~$m$ goes to infinity.  Examples of this include work by Tao and Vu, who show that for general unstructured matrices the distribution of eigenvalues divided by $\sqrt{m}$ is asymptotically uniform on the unit disk~\cite{tao2006random,tao2017random}. This is \textsl{not} true for structured matrices; see e.g.~\cite{Corless:2021:WhatCan,Corless2021} where skew-symmetric tridiagonal matrices (independently of dimension, so no scaling is needed) are seen to be confined to a diamond shape. We explain that mystery in this paper, using a century-old theorem, which deserves to be better-known.

%We largely back off from the asymptotic aspect, here, and try to answer questions about these families for ``small'' to ``moderate'' dimensions.
%For some populations, computation can go quite far, if one is careful to account for symmetries and be efficient~\cite{zivkovic2006classification}; here we concentrate on using ``off-the-shelf'' technology to explore questions as far as we can.
\section{The Effect of Matrix Structure\label{sec:structures}}
The first Bohemian results were on real symmetric matrices or Hermitian matrices~\cite{Tracy1993}. We look at other structures here, in order to get another view.  We begin with complex symmetric matrices.
\subsection{Complex Symmetric Matrices}
A complex symmetric matrix $\A$ satisfies $\A^T = \A$, where ${}^T$ is the real transpose operation.  These occur, for instance in B\'ezout matrices for polynomials with complex coefficients.  Unlike Hermitian matrices, the eigenvalues of complex symmetric matrices need not be real.  Indeed, any matrix may be brought by similarity transformation to a complex symmetric matrix~\cite[Thm 4.4.9]{horn2012matrix}. In many cases this can be done by unitary similarity; see for instance, the characterizations of when this can be done, in~\cite{Liu2013}.

Here let us examine a specific complex symmetric family, with population $-1\pm i$ where $i = (0,1)$ is the square root of $-1$. At dimension $m$, such a matrix has $m(m+1)/2$ free entries, each of which can be either of $-1 \pm i$.  This gives $2^{m(m+1)/2}$ such matrices; this growth is (much) faster than exponential.  Still, examining eigenvalues of small dimension examples can tell us much. If we take $m=6$, then the number of such matrices is only $2^{21}$, slightly more than $2$ million.  The eigenvalues of all these matrices can be computed in a reasonable time, and plotted.  As depicted in Figure~\ref{fig:symmetric6}, they seem confined to a strip in the left-half plane.

We now prove that this will always be true at any dimension.

\begin{theorem}\label{thm:bent}
If the symmetric matrix $\A$ has dimension $m$, entries drawn from $-1\pm i$, and eigenvalue $\lambda$, then $-m \le \Re(\lambda) \le 0$ and $-m \le \Im(\lambda) \le m$.
\end{theorem}
\begin{proof}
Write $\A = -\E + i\M$ where $\E = \mathbf{e}\mathbf{e}^T$ is the rank-one matrix that has all $1$s, and $\M$ is symmetric and has entries only $\pm 1$.
We use the following theorem:
[Bendixon--Bromwich--Hirsch]~\cite[Fact 5, p.~16-2]{hogben2013handbook} (original references~\cite{Bendixson1902,1902,Bromwich1906})
Write $\A = \H + i\S$ where $\H = (\A + \A^*)/2$ and $\S = (\A - \A^*)/(2i)$. Both $\H$ and $\S$ are Hermitian matrices and therefore their eigenvalues are real.
Denote the eigenvalues of $\H$ by $\mu_1 \ge \mu_2 \ge \ldots \ge \mu_m$ and the eigenvalues of $\S$ by $\nu_1 \ge \nu_2 \ge \ldots \ge \nu_m$. Then all eigenvalues $\lambda$ of $\A$ lie in the box $\mu_m \le \Re(\lambda) \le \mu_1$ and $\nu_m \le \Im(\lambda) \le \nu_1$.
This theorem can be proved in several ways; see its 1951 rediscovery by Kippenhahn as translated by~\cite{FZachlin2008}, for instance, where a particularly elegant proof using the numerical range $\Phi(\A,\mathbf{x}) := \mathbf{x}^*\A\mathbf{x}$ is given.

For our Bohemian matrix $\A=-\E+i\M$, $(\A + \A^*)/2 = -\E$, while $(\A - \A^*)/(2i) = \M$. A short computation shows that the eigenvalues of $-\E$ are $-m$ and $0$ with multiplicity $m-1$, and the Gerschgorin disk theorem shows that the eigenvalues of $\M$ lie in the union of circles centred at $1$ of radius $m-1$ and centred at $-1$ of radius $m-1$.  This establishes our theorem.
\end{proof}

\begin{remark}
{We could have stated and proved that theorem with greater generality. That is, the population could equally well have been $a \pm bi$ for real numbers $a$ and $b$ and the conclusions would have been the same, apart from scaling. For simplicity of exposition, we used only this specific example.}
\end{remark}

\begin{figure}
    \centering
    \includegraphics[width=0.9\textwidth]{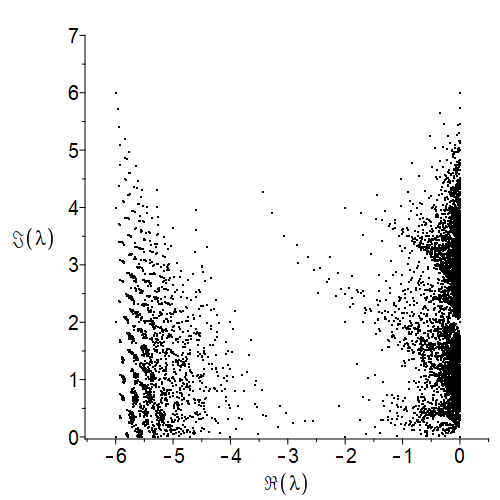}
    \caption{All eigenvalues with $\Im(\lambda) \ge 0$ of symmetric dimension $m=6$ matrices with entries $-1\pm i$. The set is symmetric about the real axis. We see the eigenvalues apparently confined to the strip $-6 \le \Re(\lambda) \le 0$, and bounded below and above by $-6 \le \Im(\lambda) \le 6$. There are several other unexplained features of this eigenvalue distribution.}
    \label{fig:symmetric6}
\end{figure}

We now use this method to explain why skew-symmetric tridiagonal matrices with population $1$ and $i$ (the population considered in~\cite{Corless:2021:WhatCan} and~\cite{Corless2021}) are confined to a diamond shape.

\subsection{Squares, Diamonds, and other shapes}
\begin{theorem}
\label{thm:square}
Let $\A$ be a square skew-symmetric matrix of dimension $m$ and population $-1 \pm i$.  Then its eigenvalues $\lambda$ satisfy $-2 \le \Re(\lambda) \le 2$ and $-2 \le \Im(\lambda) \le 2$, and are thus confined to a square.
\end{theorem}
\begin{proof}
A matrix $\A$ from this family can be written as $\A = \S + i \T$ where the superdiagonal of $\S$ is $-1$ and the superdiagonal of $\T$ is $\pm 1$. Both are skew-symmetric: $\S^T = -\S$ and $\T^T = -\T$. We have $(\A + \A^H)/2 = i\T$ and $(\A - \A^H)/(2i) = -i\S$.  Application of the Gerschgorin circle theorem to each of these shows that the eigenvalues of either matrix are confined to the interval $-2 \le \mu \le 2$.  By the theorem of Bendixon--Bromwich--Hirsch cited earlier, the eigenvalues of $\A$ are confined to the \textsl{square} $-2 \le \Re(\lambda) \le 2$, $-2 \le \Im(\lambda) \le 2$.
\end{proof}

%\textbf{Corollary}:
\begin{remark}{
Skew-symmetric tridiagonal matrices with population $1$ and $i$ are confined to a diamond $|\Re(\lambda)|+|\Im(\lambda)|\le \sqrt{2}$.  To see this, multiply the matrix by $-1+i$, which rotates its eigenvalues by $\pi/4$ and stretches them by $\sqrt{2}$; but now the matrix population is $-1\pm i$ and it's still skew-symmetric, and hence confined to the square as described above.  Rotate the square back by $\pi/4$ and shrink by $\sqrt{2}$, and the result follows.
}\end{remark}

We have therefore explained the non-round shape of the eigenvalue distribution of these matrices.
% \subsection{Characteristic compression}
% Another interesting fact that arises in these computations is that there is significant \textsl{compression} for this family in computing characteristic polynomials: at dimension $2$ there are  $8$ different matrices possible, but only $6$ different characteristic polynomials; at dimension $3$ there are % 2, 6, 20, 90, 4970
% $20$ different characteristic polynomials shared amongst $64$ different matrices; at dimension $4$ there are $90$ different characteristic polynomials shared amongst $1024$ different matrices; at dimension $5$ there are $538$ different characteristic polynomials shared amongst $32,768$ matrices.  The largest exhaustive computation we did with this family was for $m=6$, where there were only $4,970$ characteristic polynomials shared among $2,097,152$ matrices. Clearly the growth rate of the number of polynomials is significantly slower than the growth rate of the number of matrices.  We do not have any method for computing all the characteristic polynomials of symmetric matrices without first computing all the different matrices. Such would be of clear value for this family.
\subsection{Upper Hessenberg Matrices}
An upper Hessenberg matrix is a matrix of the form
\begin{equation}
    \begin{bmatrix}
    h_{1,1} & h_{1,2} & h_{1,3} & \ldots & h_{1,m} \\
    h_{2,1} & h_{2,2} &  \ddots & & \vdots \\
    0 & h_{3,2} & h_{3,3} & \ddots &  \vdots \\
    \vdots & & \ddots & \ddots & \vdots \\
    0 & \ldots & 0 & u_{m,m-1} & u_{m,m}
    \end{bmatrix}\>.
\end{equation}
That is, it is zero below the first subdiagonal.  If any entry of the first subdiagonal is zero, then the matrix is said to be reducible, because the matrix then separates into blocks containing distinct eigenvalues. Here we restrict our attention to irreducible matrices and indeed we specify that all the subdiagonal entries $h_{j+1,j} = -1$.  If the subdiagonal entries all have $|h_{j+1,j}| = 1$ we say that the matrix is \textsl{unit} upper Hessenberg.  If all the diagonal entries are zero, we say that it is Zero Diagonal.

% \begin{remark}{
% Suppose that all $h_{i,j}$ are roots of unity.  Consider the $5$ by $5$ case, for definiteness. Take four roots of unity, as yet unspecified: call them $s_2$, $s_3$, $s_4$, $s_5$. Form the diagonal matrix $\D = \diag(1,s_2, s_3, s_4, s_5)$ and perform the similarity transform $\D \H \D^{-1}$.  The result is
% \begin{equation}
%     \left[\begin{array}{ccccc}
% h_{1,1} & \frac{h_{1,2}}{s_{2}} & \frac{h_{1,3}}{s_{3}} & \frac{h_{1,4}}{s_{4}} & \frac{h_{1,5}}{s_{5}}
% \\
%  s_{2} h_{2,1} & h_{2,2} & \frac{s_{2} h_{2,3}}{s_{3}} & \frac{s_{2} h_{2,4}}{s_{4}} & \frac{s_{2} h_{2,5}}{s_{5}}
% \\
%  0 & \frac{s_{3} h_{3,2}}{s_{2}} & h_{3,3} & \frac{s_{3} h_{3,4}}{s_{4}} & \frac{s_{3} h_{3,5}}{s_{5}}
% \\
%  0 & 0 & \frac{s_{4} h_{4,3}}{s_{3}} & h_{4,4} & \frac{s_{4} h_{4,5}}{s_{5}}
% \\
%  0 & 0 & 0 & \frac{s_{5} h_{5,4}}{s_{4}} & h_{5,5}
% \end{array}\right]\>.
% \end{equation}
% Now choose $s_2 = \overline{h_{2,1}}$, $s_3 = s_2\overline{h_{3,2}}$, $s_4 = s_3 \overline{h_{4,3}}$, and $s_5 = s_4 \overline{h_{5,4}}$.  This forces the subdiagonal entries to be $1$, leaves the diagonal as it was before, and shuffles the upper triangle to be possibly different roots of unity to what they were before.  Since the elements of the upper triangle were independent, the resulting formulae will range over the entire set of possibilities as we make the $h_{i,j}$ range over the entire set of possibilities.  It is for this specific population that we started studying \textsl{unit} upper Hessenberg matrices.}
% \end{remark}

Many populations have zero as their mean value (e.g. $\{-1,1\}$ or indeed roots of unity; or $\{-1,0,1\}$).  In that case, the eigenvalues are typically symmetric about zero and a simplified picture is obtained simply by setting all the diagonal entries to zero, in which case the Gerschgorin circles are all centred at $0$.  Sometimes we will have to transform back to the nonzero diagonal case, but a surprising amount of information is retained even with the simplification of insisting on a zero diagonal.

%Upper Hessenberg matrices occur in the numerical solution of eigenvalue problems via $QR$ iteration as an intermediate step: every matrix $\mathbf{A}$ is unitarily similar to an upper Hessenberg matrix, and this conversion can occur in $O(n^3)$ operations.  Subsequent $QR$ iterations are then much cheaper.  These matrices also occur in other contexts (and indeed have been known to be reinvented).  For us, they become an object of study in and of themselves; admittedly, we remain chiefly interested in eigenvalues.
\begin{figure}
    \centering
    \includegraphics[width=0.93\textwidth]{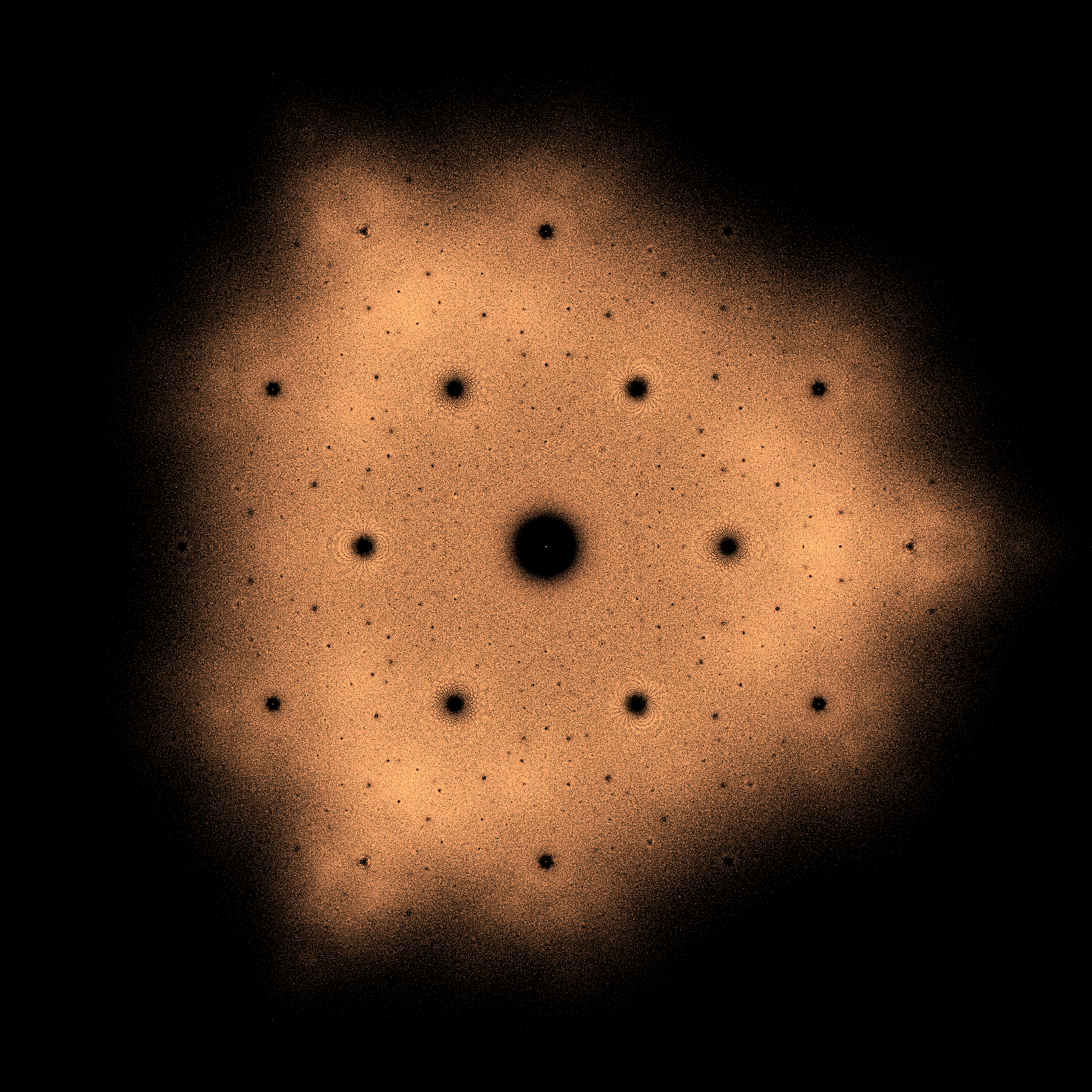}
    \caption{Eigenvalues of a sample of $5$ million upper Hessenberg matrices of dimension $m=5$ matrices with population cube roots of unity. The image is visually indistinguishable from the density plot of all $14,348,907$ unit upper Hessenberg zero diagonal matrices of dimension $m=6$.}
    \label{fig:sampledupperHessenberg}
\end{figure}
\subsection{Rayleigh Quotients}
Recall the Rayleigh Quotient:
\begin{equation}
    r = \frac{\mathbf{y}^T \A \mathbf{x}}{\mathbf{y}^T \mathbf{x}}\>.
\end{equation}
If $\mathbf{x}$ is an approximate eigenvector, and $\mathbf{y}^T$ an approximate left eigenvector, then this quotient is a least-squares approximation to an eigenvalue of $\A$.  If we replace $\A$ by $\A^{-1}$ above, then this is an approximation to an eigenvalue of $\A^{-1}$, and typically the largest one; this of course is the reciprocal of the smallest eigenvalue of $\A$.

We will consider this not as an eigenvalue approximation, but as a process in its own right, and plot the results of a single iteration of this on a Bohemian family, with both $\mathbf{y}$ and $\mathbf{x}$ taken to be the first elementary vectors. Thus the result is the top left corner of the inverse of our Bohemian matrix.  See Figure~\ref{fig:Rayleigh}.
\begin{figure}
    \centering
    \includegraphics[width=0.9\textwidth]{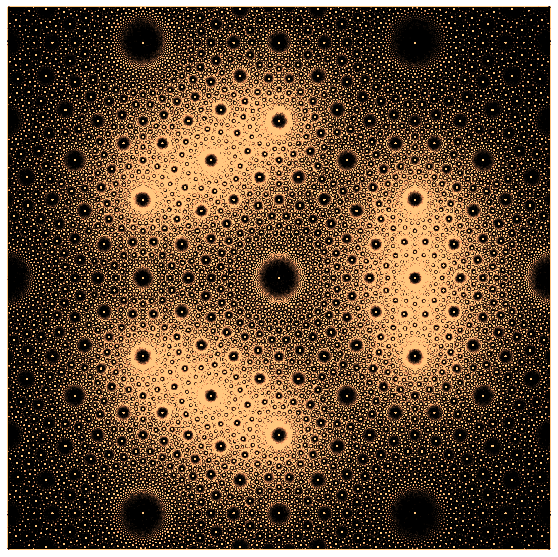}
    \caption{Density plot of the upper left corner element of $\A^{-1}$ where $\A$ is sampled randomly from the $8$ by $8$  upper Hessenberg Bohemian (\textsl{not} unit or zero diagonal) family with population cube roots of unity.  There are $3^{43} > 3.2\cdot 10^{20}$ such matrices; we sampled only $2\cdot 10^7$ of these and gave a density plot on a $2048$ by $2048$ grid, enhanced by anti-aliased point rendering.}
    \label{fig:Rayleigh}
\end{figure}
This can also be computed by the recurrence relation~\eqref{eq:recur}, as the ratio of two determinants, from Cramer's rule:
\begin{equation}
    r = \frac{Q_{m}(0;t_1,t_2,\ldots,t_{m-1})}{Q_{m-1}(0;t_1,t_2,\ldots,t_{m-2})}
\end{equation}
Perhaps surprisingly to a numerical analyst, this offers an effective way to perform this computation when the population consists of small Gaussian integers, which can be represented as complex ``flints'' and are not subject to rounding error when ring arithmetic is carried out in floats. In this case the only division occurs at the end, and so rounding error is trivial.  Even with roots of unity, the rounding errors are generally not of serious consequence.

What we are computing here is representable in (if cube roots of unity are used) $\mathbb{Q}(\sqrt{-3})$, and moreover the rational numbers involved will not have overly large values.  This suggests that the appearance in Figure~\ref{fig:Rayleigh} can be explained as Schmidt arrangements, as done by Katherine Stange.  See her blog at~\url{https://math.katestange.net/}, but note in particular~\cite{Stange2017} and her papers previously cited.

% \begin{figure}
%     \centering
%     \includegraphics[width=0.45\textwidth]{Exhaustivepop_CubeRootsOfUnity_copper_6N14348907.png}
%     \caption{Eigenvalues of a all $14,348,907$ unit upper Hessenberg zero diagonal matrices of dimension $m=6$ matrices with population cube roots of unity. In contrast to Figure~\ref{fig:sampledupperHessenberg}, the dimension of the matrix is higher, but all subdiagonal entries are $1$ and all diagonal entries are $0$.  The picture is remarkably similar anyway.}
%     \label{fig:exhaustiveunitupperHessenbergZD}
% \end{figure}
\subsection{Unit upper Hessenberg and Toeplitz zero diagonal matrices}
We will use the following Theorem repeatedly, to ensure that we see all the eigenvalues of the family in the window of the plot.

\begin{theorem}\label{thm:bound}
Suppose that every entry of a unit upper Hessenberg zero diagonal matrix $\mathbf{H}$ has magnitude at most $B$: that is, $|h_{i,j}| \le B$.  Then every eigenvalue $\lambda$ of $\mathbf{H}$ is bounded independently of dimension by
\begin{equation}
    |\lambda| \le 1 + 2\sqrt{B}\>.
\end{equation}
\end{theorem}
\begin{proof}
The proof uses an idea already present in~\cite{Schmidt1960}: namely, one chooses a diagonal matrix $\D = \diag(1, r, r^2, \ldots, r^{m-1})$ with a free parameter $r>1$ and considers the similar matrix $\D\H\D^{-1}$, which therefore has the same eigenvalues as $\H$.  Suppose, without loss of generality, that the subdiagonal entries of the zero diagonal unit upper Hesssenberg $\H$ are all $-1$.  Then
\begin{equation}
\D\H\D^{-1} =     \begin{bmatrix}
    0 & h_{1,2}/r & h_{1,3}/r^2 & \ldots & h_{1,m}/r^{m-1} \\
    -r & 0 &  h_{2,3}/r & & h_{2,m}/r^{m-2} \\
    0 & -r & 0 & h_{3,4}/r &  \vdots \\
    \vdots & & \ddots & \ddots & \vdots \\
    0 & \ldots & 0 & -r & 0
    \end{bmatrix}\>.
\end{equation}
The Gerschgorin disk for the first row has radius at most $B/(r-1)$ by comparison with a geometric series; the second and all subsequent rows have the bound $r + B/(r-1)$ for the radius which is larger because $r>1$.  To minimize this bound, we write it as $1 + r-1 + B/(r-1)$ and use the AGM inequality to say that this is minimized when $r-1 = B/(r-1)$ or $r = 1+\sqrt{B}$; this gives the value of the Gershgorin radius as $1+2\sqrt{B}$, as desired.
\end{proof}
\begin{remark}{
This is the only Gerschgorin-like theorem that we are aware of that gives a bound for eigenvalues which is independent of the dimension and depends on the \textsl{square root} of the bound for the entries in the matrix instead of the more usual linear power of the bound.  Of course, if one multiplies a matrix $\A$ by a constant factor, then the eigenvalues must also be multiplied by that factor; but we cannot perform such a multiplication here and remain in the class of \textsl{unit} upper Hessenberg matrices.}
\end{remark}

%\subsection{Unit upper Hessenberg zero diagonal Toeplitz Matrices}
A Toeplitz matrix $\T$ has constant elements on every diagonal, that is, $t_{i,j} = t_{0,j-i}$. The authors of~\cite{Chan2020} found that the  Toeplitz subset of unit upper Hessenberg matrices maximized the \textsl{characteristic height}, and so decided to study that subset directly; it contains only exponentially many elements ($\#P^{m-1}$, rather than $\#P^{O(m^2)}$) and has several other interesting features.  For large dimension, the spectral theory connects to the well-known asymptotic spectral theory for Toeplitz matrices.
\begin{remark}
``Well-known'' is not the same as ``everyone knows". The major results of this area include the surprising fact that eigenvalues of finite-dimensional Toeplitz matrices do \emph{not} approach the eigenvalues of the infinite-dimensional Toeplitz operator, although the \emph{pseudospectra} do~\cite{Trefethen2005,corless2007pseudospectra,corless2013}. We shall not be concerned with pseudospectra in this paper, although they play a role for computation of eigenvalues of even modestly large dimension Toeplitz matrices.
\end{remark}
\begin{figure}
    \centering
    \includegraphics[width=0.9\textwidth]{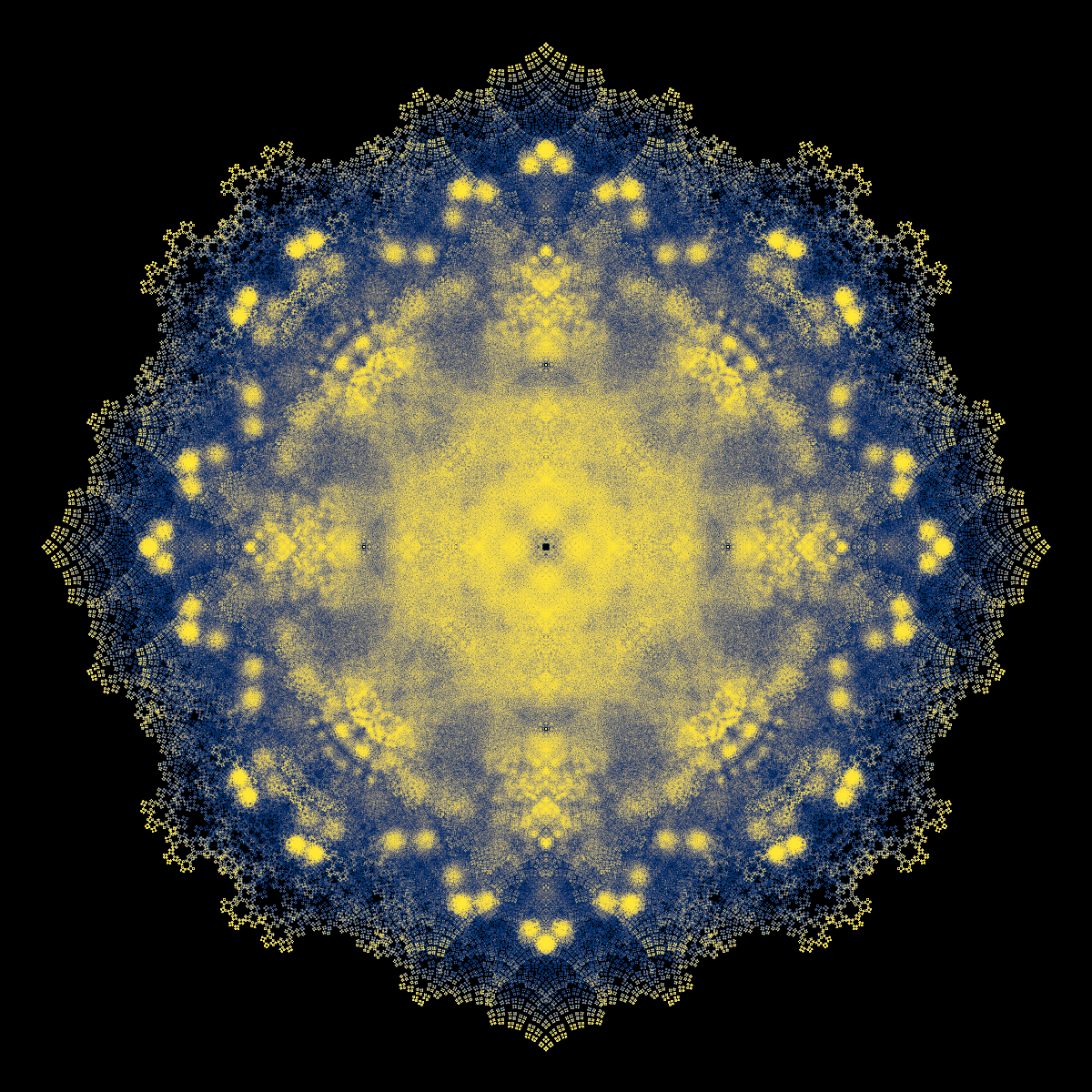}
    \caption{Density plot of  eigenvalues of all $1,048,576$ upper Hessenberg Toeplitz matrices of dimension $11$ with zero diagonal, $-1$ subdiagonal, and population $\pm 1$, $\pm i$ (fourth roots of unity) otherwise.  Brighter colours correspond to higher density.}
    \label{fig:fourthroots}
\end{figure}
\begin{figure}
    \centering
    \includegraphics[width=0.9\textwidth]{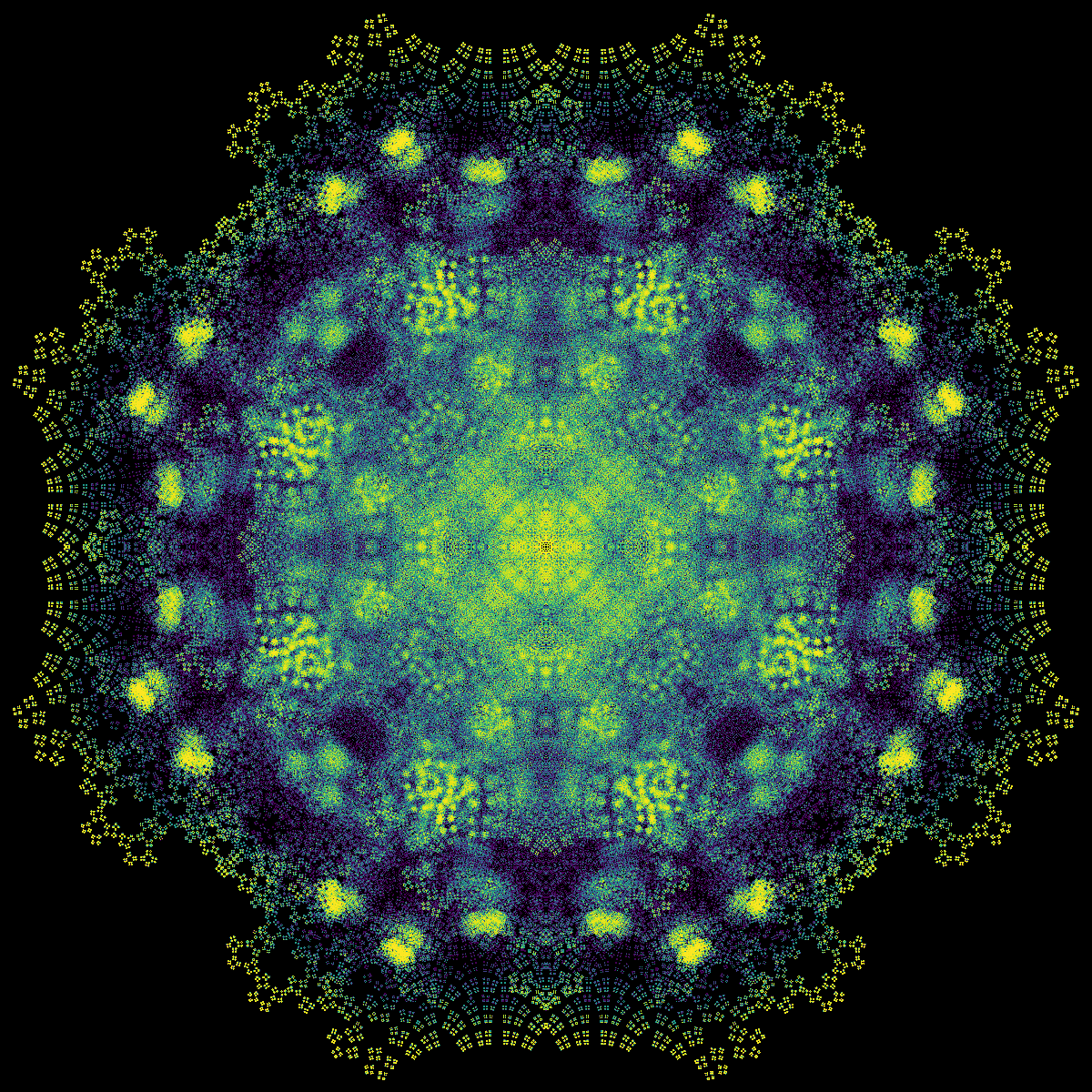}
    \caption{Density plot of  eigenvalues of all $262,144$ upper Hessenberg Toeplitz matrices of dimension $10$ with zero diagonal, $-1$ subdiagonal, and population $\pm 1 \pm i$ (four corners of a square) otherwise.  Brighter colours correspond to higher density.}
    \label{fig:scottishcross}
\end{figure}

%\subsection{The Fractal Edges}
At the edges of Figures~\ref{fig:fourthroots} and~\ref{fig:scottishcross} we see clear indication of fractal gasket-like structures. When the population is third roots of unity, we see Sierpinski gaskets; when the population only has two elements, we see pairs, and pairs of pairs, recursively; with five-element populations we see recursive pentagonal structures.  The phenomenon seems universal for upper Hessenberg Toeplitz matrices. We conjectured that there must be a recursive construction underlying this, and we now know this to be true.

The keys to understanding this are the spectral theory of Toeplitz matrices, the improved Gerschgorin bound of Theorem~\ref{thm:bound}, and a recurrence for the characteristic polynomials.  Denote the characteristic polynomial of a unit upper Hessenberg zero diagonal matrix by $Q$: then we have
\begin{align}
    Q_{n+1}(z; t_1, t_2, \ldots, t_{n})
    =& z Q_n(z; t_1, t_2, \ldots, t_{n-1}) \nonumber \\
    &{}- \sum_{k=1}^{n} (-1)^k t_k Q_{n-k}(z; t_1, t_2, \ldots, t_{n-k-1})\>.\label{eq:recur}
\end{align}
There is a somewhat more complicated recurrence relation for a general upper Hessenberg matrix; see~\cite{Chan2020}.
Proofs can be found in many places, for instance~\cite{Cahill2002}.
The final term is $\pm t_n Q_0(z)$ and $Q_0(z) = 1$.

Using Theorem~\ref{thm:bound} with $r=2$ (this is also used in~\cite{bogoya2022upper}) we find the Toeplitz symbol to be
\begin{equation}
    -\frac{2}{e^{i\theta}} + \sum_{k\ge 1} \frac{t_k}{2^k}e^{ik\theta}\>.
\end{equation}

The eigenvalues of any $m$-dimensional Toeplitz matrix with a \emph{finite} symbol---that is, a Laurent polynomial---are known to converge as $m\to\infty$ to the union of several algebraic curves described by the Schmidt--Spitzer theorem~\cite[Ch.~22]{hogben2013handbook}. For brevity we will write SS for Schmidt--Spitzer hereafter. An improved algorithm for computing these curves is given in~\cite{Bottcher2021}, which we have implemented. For convenience, we state the SS theorem here.

\begin{theorem}
Given the Laurent polynomial
\begin{equation}
    f(z) = \sum_{\nu=-q}^h a_{\nu} z^\nu
\end{equation}
where $q>0$, $h>0$ and $a_{-q}a_h \ne 0$, and given $\lambda \in \mathbb{C}$, define the polynomial
\begin{equation}
    Q(\lambda;z) = z^q(f(z) - \lambda)\>.
\end{equation}
Denote the moduli of its $q+h$ zeros counted according to multiplicity by $\alpha_j$ for $1 \le j \le q+h$.  Assume that they have been ordered so that $\alpha_1 \le \alpha_2 \le \cdots \le \alpha_{q+h}$.  Then $\lambda \in \Lambda$, that is the set of SS points, if and only if the $q$th largest and $(q+1)$st largest roots have the same magnitude: $\alpha_q = \alpha_{q+1}$. Recall that $q$ is fixed by the Laurent polynomial.
\end{theorem}
\begin{proof}
See~\cite{Schmidt1960}. One key piece is to use Rouch\'e's theorem and a contour integral for a winding number; the polynomial nature of the symbol is essential. The behaviour of the nonpolynomial case is usually quite different, and connect with pseudospectra.  The proofs in that paper also typically assume that eigenvalues are simple, which is generically true.  For us, multiple eigenvalues have nonzero but typically exponentially small probability. Extensions to this theorem such as in~\cite{garoni2017generalized} use more sophisticated ideas.
\end{proof}
% Referee wants an idea of proof

What we observed experimentally, as demonstrated in Figure~\ref{fig:SchmidtSpitzer}, is that the SS curves themselves seem to converge rapidly (basically already for dimension $n=m$) as we add new sequence elements $t_k$. This is \emph{in spite} of the known sensitivity of eigenvalues of Toeplitz matrices to perturbations; indeed the sensitivity is exponential in the dimension of the matrix. Why, then, do we see convergence, in this case?

The following sequence of lemmas explain how this can be true.
\begin{lemma}
Suppose a fixed infinite sequence $t_k \in P$ is given, where each $|t_k|\le B$. Take $\rho \in (0,1)$.
The truncated Toeplitz symbols $a_{m}(z) = -1/z + \sum_{k=1}^{m-1}t_k z^k$ converge in $0 < |z|\le \rho <1$ to a mereomorphic function $a(z)$, indeed analytic in the punctured disk $0 < |z| \le \rho < 1$.
\end{lemma}
\begin{proof}
Immediate, because the tail is bounded by a geometric series.
\end{proof}
\begin{lemma}
The scaled symbols $a_{m}(z/r) = -r/z + \sum_{k=1}^{m-1} t_k/r^k z^k$ converge to a function analytic in the punctured disk $0 < |z| \le \rho < r$, where now $\rho$ can be taken to be larger than $1$ by taking $r>1$. Moreover, the SS curves of all the unscaled symbols $a_m(z)$ are contained in the images $a_{m}(e^{i\theta}/r)$ of the unit circle, which are closed curves contained in a finite subset of the complex plane. Therefore the SS curves are also contained in the intersection of these images over all $r>1$.
\end{lemma}
\begin{proof}
Almost immediate, by using the same diagonal scaling used in the proof of Theorem~\ref{thm:bound}.
\end{proof}
\begin{remark}
It is this scaling using $r>1$ which allows us to use symbols which are not in $L^\infty$ when $r=1$: the eigenvalues and curves which the \emph{scaled} and therefore $L^\infty$ symbols enclose are precisely the same eigenvalues and curves as the unscaled version.
Our use of this scaling requires the matrix to have only finitely many subdiagonals, such as being upper Hessenberg.
\end{remark}
\begin{theorem}
The SS curves for the unscaled $a_{m}(z)$ converge in the Hausdorff metric,  in the limit as $m\to \infty$,  to a set $\Lambda$ of piecewise analytic arcs inside the image $a(e^{i\theta}/r)$ for any $r>1$.
\end{theorem}
\begin{proof}
The Laurent polynomial symbols $a_{m}(z)$ obtained by truncation to degree $m-1$ converge uniformly to an analytic function $a(z)$ as $m\to \infty$ in $0 < |z| \le \rho < 1$, by Theorem 2.7a in~\cite[Vol I]{Henrici}.  By \href{https://en.wikipedia.org/wiki/Hurwitz's_theorem_(complex_analysis)}{Hurwitz' theorem}~\cite{titchmarsh1939theory} for fixed $\theta \ne 0$ the zeros of the sequence $a_{m}(z)-a_m(ze^{i\theta})$ also converge to a (possibly multiple) zero of $a(z)-a(ze^{i\theta})$, for any fixed $\theta$, and every zero of $a(z)-a(ze^{i\theta})$ has the appropriate number of zeros of the sequence converge to it.  If these two equal-magnitude roots, which are distinct because $\theta\ne 0$ and both in $0 <|z| \le \rho < 1$, are the smallest (which are what we care about, because $q=1$), then $\lambda = a(z) = a(ze^{i\theta}) \in \Lambda$, the SS curve.
\end{proof}
\begin{remark}
This theorem partially explains why eigenvalues of the $m$-dimensional Toeplitz matrices with symbol $a_{m}(z)$ converge to $\Lambda$ as $m \to \infty$.
The symbols $a_m(z)$ converge to $a(z)$, and the number of roots of $a(z)-a(ze^{i\theta})$ in any punctured disk $0 < |z| \le \rho < 1$ is finite by the regularity of $a(z)$; therefore, there exists $m_0$ so so that the number of roots of $a_m(z)-a_m(ze^{i\theta})$ if $m\ge m_0$ will be at least $1$ because the \emph{finite} SS curve (with all $t_k = 0$ if $k\ge m_0$) is not empty; call the smallest magnitude root $z_m^*$ and the corresponding limiting root $z_\infty^*$. Since, by compactness, the roots  converge uniformly as $m \to \infty$, $\lambda^* = a(z_\infty^*)$ is close to $\lambda_m^* = a_m(z_m^*)$.  Therefore, the the eigenvalues of \emph{any} Toeplitz family where we \emph{truncate} to a Laurent polynomial of degree $n > m$ will be close to the same curve.  Yet this is only a partial explanation: above, we remarked that the \emph{eigenvalues} are close to the SS curves already for dimension $n=m$, when the SS theorem only guarantees that as $n\to \infty$ the eigenvalues will cluster around the SS curves. It is absolutely critical that these perturbations occur in the top right corner of the matrix; perturbations in the bottom left corner will have drastic effects on the eigenvalues, and the resulting pictures are then explainable by the theory of pseudospectra.
\end{remark}

\begin{remark}{
The SS theorem entails that, for a Laurent polynomial with a pole of order $q$, it is the zeros of $q$th smallest magnitude that matter. For our application, $q=1$ because the matrix is not only Toeplitz but also unit upper Hessenberg with zero diagonal. So it is the \emph{smallest} magnitude zeros that matter, and these are the first ones to converge, in practice.

Consider the example\footnote{This sequence does not give an $L^\infty$ $a(z)$, but that does not matter because we can rescale to $t_k/r^k$, as previously noted.} where all $t_k = 1$. The symbol is then, when $r=1$,}
\begin{equation}
    a(z) = -\frac{1}{z} + \sum_{k\ge 1} z^k
\end{equation}
{or $-1/z + z/(1-z)$. Only the smallest two roots (because the symbol is real, there are two which are complex conjugates of each other) of $a_m(z)-a_m(ze^{i\theta})$ lie inside the image of the unit circle, and the other polynomial zeros in some sense approach the circle of convergence $|z|=1$, as we can see in this case by direct computation and as we know in general by a famous result of Jentsch; for details and an extension of that famous result see~\cite{Blatt2009}.  The compactness, and the scaling by $r>1$,
really matter.}

{Applying the improved algorithm of~\cite{Bottcher2021}, we solve $a(z) - a(e^{i\theta}z)$ and find the two zeros
\begin{equation}
    z = \frac{{\mathrm e}^{\mathrm{i} \theta}+1\pm\sqrt{\left({\mathrm e}^{\mathrm{i} \theta}\right)^{2}-6 \,{\mathrm e}^{\mathrm{i} \theta}+1}}{4 \,{\mathrm e}^{\mathrm{i} \theta}}\>.
\end{equation}
It is not obvious, but $|z| = 1/\sqrt{2} < 1$ for both of these zeros; indeed $z =  e^{\pm i\phi}/\sqrt{2}$ where $\phi$ varies from $\pi/4$ to $\pi$ as $\theta$ varies from $0$ to $\pi$.
Drawing $\lambda = a(z)$ in the complex plane for $-\pi \le \theta \le \pi$ gives the correct half-circle of radius $2$ centred at $-1$ which is the curve $\Lambda$ describing the asymptotic location of the eigenvalues of the Toeplitz matrices with this symbol.}
\end{remark}
% \begin{remark}{
% We have not investigated the case $q>1$, that is, Toeplitz matrices with more subdiagonals.  We suspect that the above theorem will go through, \emph{mutatis mutandis}, perhaps with slower convergence the more nonzero subdiagonals the matrices have.}
% \end{remark}

% 1,-1,1,-1,-1,-1,-1,1,1,1,-1,1,
\begin{figure}
  \centering

  \begin{subfigure}[t]{.32\linewidth}
    \centering\includegraphics[width=.925\linewidth]{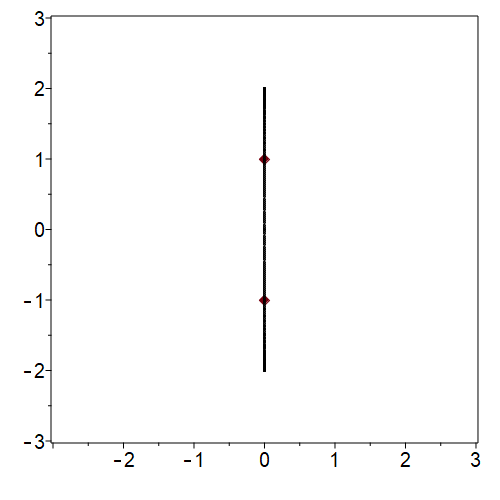}
    \caption{$t_1 = 1$}
  \end{subfigure}
  \begin{subfigure}[t]{.32\linewidth}
    \centering\includegraphics[width=.925\linewidth]{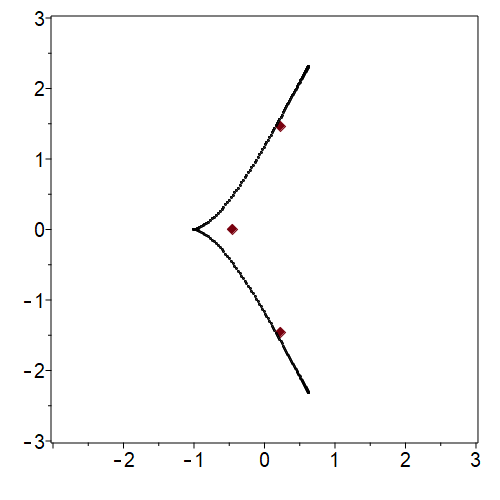}
    \caption{$t_1=1$, $t_2=-1$}
  \end{subfigure}
  \begin{subfigure}[t]{.32\linewidth}
    \centering\includegraphics[width=.925\linewidth]{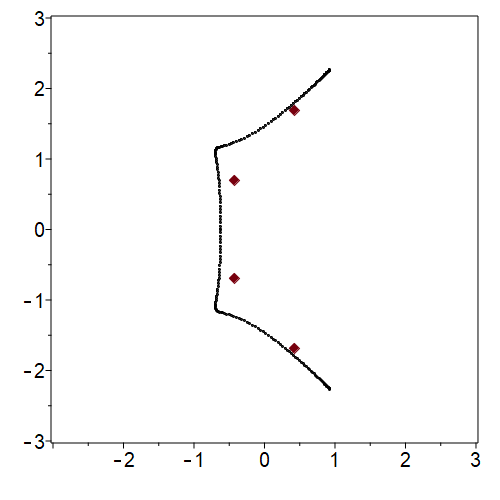}
    \caption{$t_3=1$}
  \end{subfigure}

  \begin{subfigure}[t]{.32\linewidth}
    \centering\includegraphics[width=.925\linewidth]{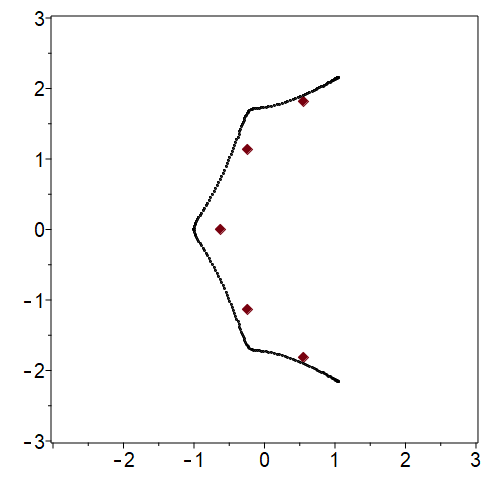}
    \caption{$t_4=-1$}
  \end{subfigure}
  \begin{subfigure}[t]{.32\linewidth}
    \centering\includegraphics[width=.925\linewidth]{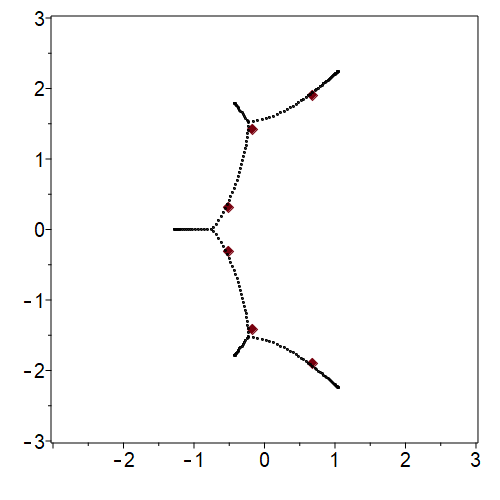}
    \caption{$t_5=-1$}
  \end{subfigure}
  \begin{subfigure}[t]{.32\linewidth}
    \centering\includegraphics[width=.925\linewidth]{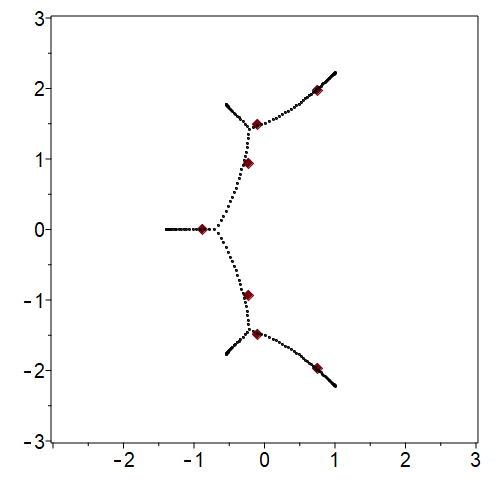}
    \caption{$t_6=-1$}
  \end{subfigure}

  \begin{subfigure}[t]{.32\linewidth}
    \centering\includegraphics[width=.925\linewidth]{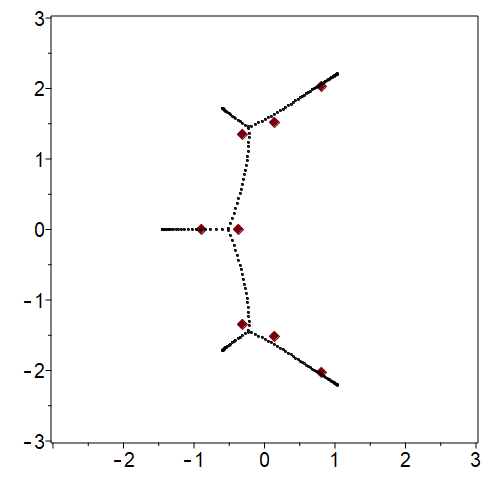}
    \caption{$t_7=-1$}
  \end{subfigure}
  \begin{subfigure}[t]{.32\linewidth}
    \centering\includegraphics[width=.925\linewidth]{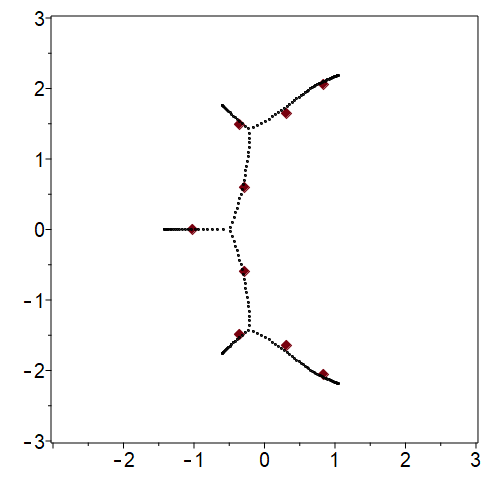}
    \caption{$t_8=1$}
  \end{subfigure}
  \begin{subfigure}[t]{.32\linewidth}
    \centering\includegraphics[width=.925\linewidth]{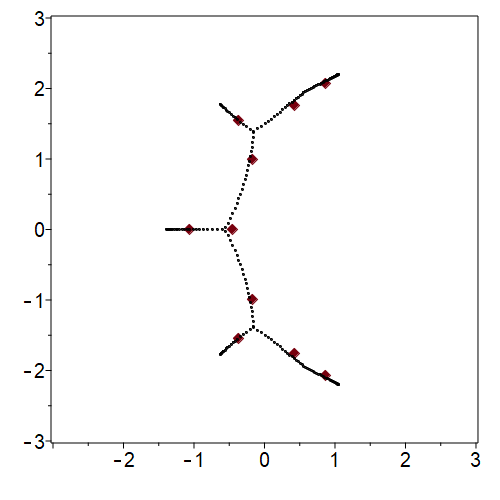}
    \caption{$t_9=1$}
  \end{subfigure}
 \caption{\label{fig:SchmidtSpitzer}The SS curves for Toeplitz matrices with symbol $a(z) = -1/z + \sum_{k=1}^{m-1} t_k z^k$ for a fixed sequence of $t_k$ with $|t_k|\le 1$ and various $m$. As we increase $m$, and thus see a new member of the sequence, the curves are seen to converge quite rapidly. After $m=10$ they are visually indistinguishable from the final case shown here. The red dots are the eigenvalues of the dimension $m$ Toeplitz matrix with those entries. In the limit as $m\to\infty$ these are guaranteed to converge to the SS curves but we see that they are visually there from nearly the beginning.}
\end{figure}

% For each of the three choices of $t_n$, the resulting characteristic polynomial $Q_{n+1}(z)$ can be written as a \textsl{fixed} polynomial $zQ_n(z) - \sum_{k=1}^{n-1} (-1)^k t_k Q_{n-k}$ plus $t_n$.
\subsection{Explanation of the Fractal Appearance}
Suppose for instance that the population $P$ of a Bohemian upper Hessenberg zero diagonal family has three distinct elements.  Then
\begin{equation}
    Q_{n+1}(z) = F_n(z) + t_n
\end{equation}
where $F_n(z) = zQ_n(z) - \sum_{k=1}^{n-1} (-1)^k t_k Q_{n-k}$ is a fixed polynomial depending only on previous $t_k$.
This final term $t_n$ perturbs that fixed polynomial in one of (in this case) three ways.  Now use a homotopy argument: Replacing $t_n$ by $st_n$ where $0 \le s \le 1$, we see the roots of $Q_{n+1}$ arising by paths emanating in three directions from the roots of $F_n(z)$.
That is, \textsl{for each root of $F_n(z)$, three nearby roots of different $Q_{n+1}(z)$ arise}.
The roots of $F_n(z)$ are also near the zeros of the symbol.  This explains the fractal structure.

This recursive construction is not a linear one: the structures resembling Sierpinski gaskets seen in the close-up in Figure~\ref{fig:Sierpinski} are clearly not rigid triangles, but rather have been distorted into curved shapes. Nonetheless we believe the above explanation is one way of understanding why this structure arises.
\begin{figure}
    \centering
    \includegraphics[width=0.9\textwidth]{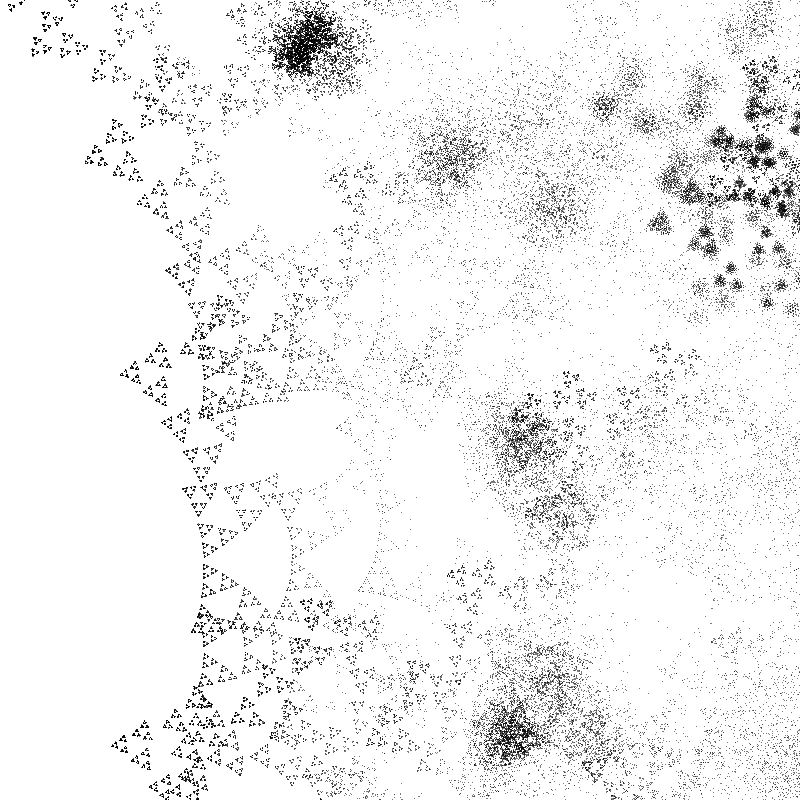}
    \caption{A close-up (window $-2.5 \le \Re(\lambda) \le -1.5$, $0 \le \Im(\lambda) \le 1$) of an $800$ by $800$ density plot of the eigenvalues of all $531,441$ upper Hessenberg zero diagonal matrices with population cube roots of unity.  The resemblance to a Sierpinski gasket is striking.}
    \label{fig:Sierpinski}
\end{figure}

\section{On Visualization\label{sec:visualization}}
Most of the figures in this present paper use only the simplest techniques of visualization: Figures~\ref{fig:uniform} and~\ref{fig:diamond} show \textsl{colourized} density plots of eigenvalues in the complex plane, about which more in a moment. Figure~\ref{fig:symmetric6} is a simple plot.  Figure~\ref{fig:Sierpinski} is a greyscale density plot on an $800$ by $800$ grid.  Figures~\ref{fig:fourthroots} and~\ref{fig:scottishcross} are colourized  density plots, where the colouring scheme was chosen by using a cumulative frequency count in order to attempt to equalize the \textsl{apparent density} of eigenvalues using colour; this verges on true computer imaging techniques but is actually very crude.  The technique has some value because it is relatively faithful to the underlying mathematics: brighter colours correspond to higher eigenvalue density, and when the viridis or, better, cividis colour palette is used, the colours are relatively even perceptually~\cite{nunez2018optimizing}.  The copper palette of Figure~\ref{fig:sampledupperHessenberg}
%and~\ref{fig:exhaustiveunitupperHessenbergZD}
retains the correlation of brightness to density, but has a smaller colour range.

The appearance of Figure~\ref{fig:Rayleigh}, however, depends on some rather more professional techniques, as described in~\cite{Piponi2009}.  The basic idea is to estimate spatial derivative information (using TensorFlow Gradient) and use that to enhance the figure, making the density visible even with relatively sparse data (for this figure, only $500,000$ matrices were used, and the computational cost was substantially lower than for the other figures).  The technique is called ``\href{https://en.wikipedia.org/wiki/Spatial_anti-aliasing#Signal_processing_approach_to_anti-aliasing}{anti-aliased point rendering}.'' The picture remains faithful to the underlying mathematics, however.

Figure~\ref{fig:stable} also uses this enhancement, this time because without it some features of the eigenvalues (relative increase in density near the edge $\Re(\lambda)=0$, for instance) are not so easy to see.

\section{Concluding remarks\label{sec:conclusion}}
The notion of a \textsl{Bohemian matrix} seems to be a remarkably productive one, with substantial connections to very active areas of research, including visualizations in number theory, combinatorial design, random matrices in physics, numerical analysis, and computer algebra.
% In an earlier section, we asked several questions (most of which are unanswered in general) and claimed that there were infinitely many more.  We will talk here about which of these we think is worth tackling next.

Many combinatorial questions about Bohemians, such as ``how many different characteristic polynomials are there'' for a given Bohemian family (say unit upper Hessenberg with population $(-1,0,1)$ for concreteness), can be addressed computationally, but floating-point error is an issue.  One is forced to look at polynomials (and thus computer algebra, however implemented) because of the \textsl{multiple-eigenvalue} problem: in those circumstances, numerical computation of eigenvalues is ill-conditioned and one cannot really count things by ``clustering'' nearby eigenvalues.  One is often tempted, when thinking of random matrices, to say that ``multiple eigenvalues never happen'' but of course this is not true, although typically exponentially unlikely.  Supplying constraints (either on the population or the matrix structure) significantly enhances the probability that multiplicity will be encountered.

One topic some of us have looked at briefly is that of \textsl{stable} matrices.  Which Bohemian matrices have all their eigenvalues strictly in the left half-plane?  For those matrices $\A$, and those matrices only, the solutions to the linear differential equation $\dot y = \A y$ will ultimately decay to zero.  If the dimensions are large, then one may have to consider pseudospectra (and thus matrix non-normality) as well.
% A similar question asks which Bohemian matrices have all their eigenvalues inside the unit circle?  For those matrices, the solution to the linear difference equation $y_{n+1} = \A y_n$ will ultimately decay.  If the matrix is non-normal, again pseudospectra will play a role.

It can be unsatisfactory to compute the eigenvalues of a matrix $\A$ and check to see if they are all in the left half plane; rounding errors may drift some of them into the right half plane.  Computation of the characteristic polynomial, and subsequent use of the Routh--Hurwitz criterion, seems in order.  One would like to take advantage of the compression seen for several families: rather than computing eigenvalues of several million matrices, instead compute the roots of the (equivalent) several thousand characteristic polynomials. Better yet, apply the Routh--Hurwitz criterion, which is a rational criterion, to make the decision in an arena uncontaminated by rounding errors.

As an example, consider the symmetric matrices with population $-1\pm i$ of dimension $m=6$. We already know from Theorem~\ref{thm:bent} that all eigenvalues lie in $\Re(\lambda) \le 0$. But are there any of the $4,970$ characteristic polynomials of these $2^{21}$ matrices which have \textsl{all} of their roots strictly in the left half plane? Yes.  By applying Maple's Hurwitz tool to the characteristic polynomials (which were actually computed using Python and exported in a JSON container to Maple) we identified 1328 of these polynomials, all of whose roots were strictly in $\Re( \lambda) < 0$.  Indeed, the maximum real part was approximately $-1.03\cdot 10^{-5}$.  Corresponding to these $1328$ polynomials were $966,240$ matrices, or about $46\%$ of the total.

For upper Hessenberg matrices with population $(-1,-1\pm i)$ and dimension $m=4$, out of $1,594,323$ matrices we find $365,307$ distinct characteristic polynomials.  Of these, only $14,604$ (associated with $66,782$ matrices, about $4.2\%$ of the total) have all their roots strictly in the left half plane.  The maximum real part is about $-7.1\cdot10^{-5}$.  See Figure~\ref{fig:stable}. We enhanced the figure to show more clearly the increase in density near $\Re(\lambda)=0$.  For comparison, we plot the eigenvalues of the \emph{whole} family in Figure~\ref{fig:unstableandstable}. One sees immediately that the majority of matrices in this collection have eigenvalues in the right half-plane.

\begin{figure}
    \centering
    \includegraphics[width=0.9\textwidth]{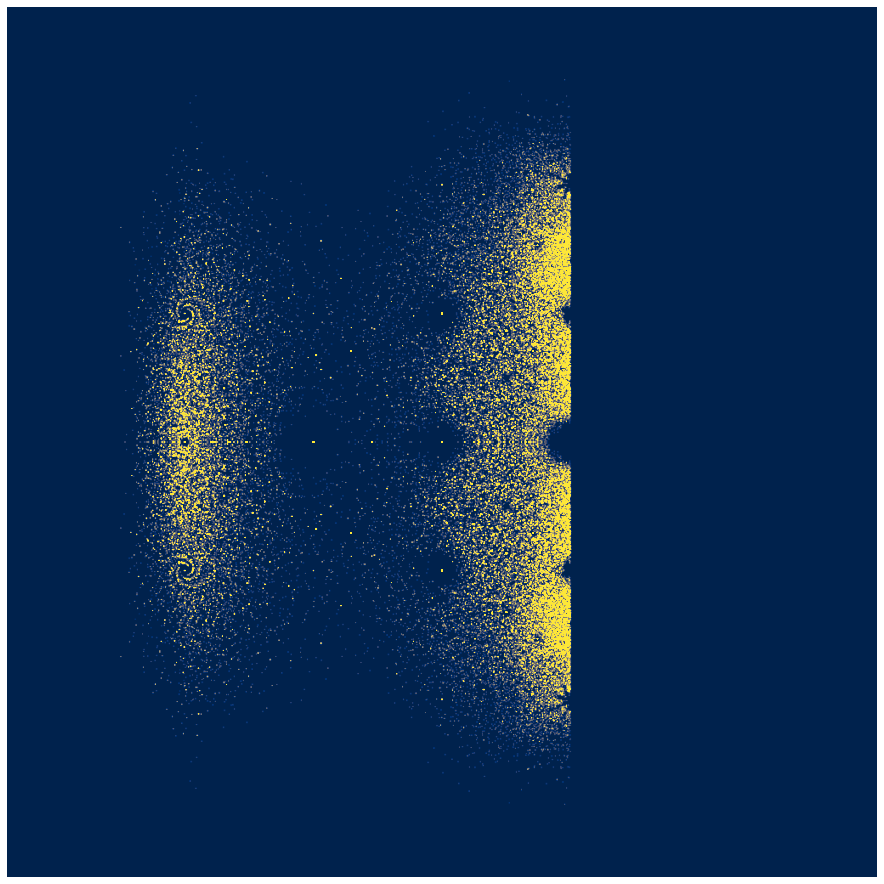}
    \caption{A $1024$ by $1024$ grid density plot of the roots of all $14,604$ stable characteristic polynomials, enhanced by anti-aliased point rendering. The Bohemian family is upper Hessenberg, population $-1-i$, $-1$, $-1+i$, and dimension $m=4$. The plot is on $-L-1 \le \Re(\lambda) < L-1$, $-L \le \Im (\lambda) \le L$, where $L = 1 + 2\cdot 2^{1/4}$. }
    \label{fig:stable}
\end{figure}

\begin{figure}
    \centering
    \includegraphics[width=0.9\textwidth]{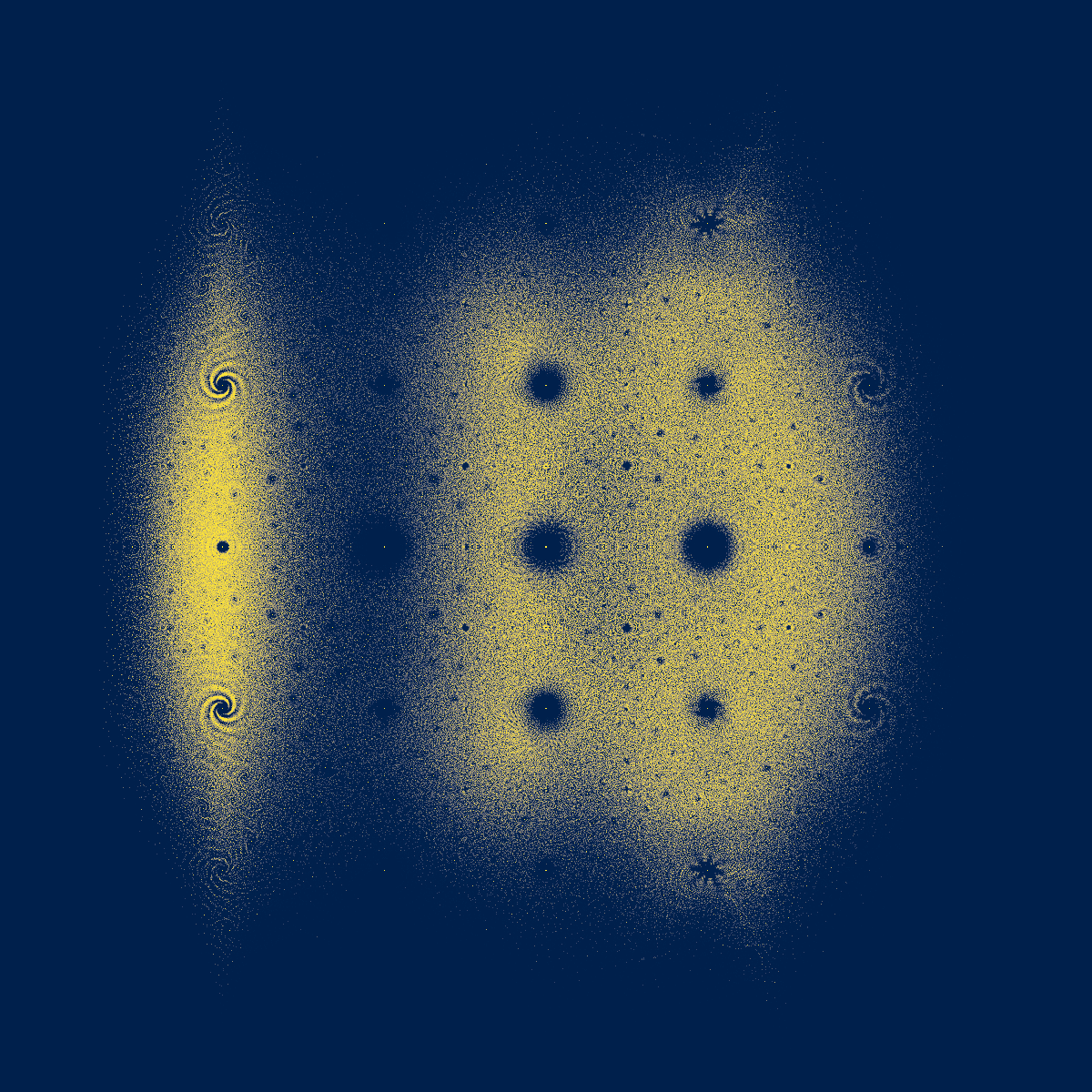}
    \caption{A $1200$ by $1200$ grid unenhanced density plot of the eigenvalues of the whole family, to compare with the plot of Figure~\ref{fig:stable}. The Bohemian family is upper Hessenberg, population $-1-i$, $-1$, $-1+i$, and dimension $m=4$. The plot is on $-L-1 \le \Re(\lambda) < L-1$, $-L \le \Im( \lambda) \le L$, where $L = 1 + 2\cdot 2^{1/4}$. The spirals are completely unexplained. }
    \label{fig:unstableandstable}
\end{figure}

One of the referees suggested that we compare our results on skew-symmetric tridiagonal matrices with those of~\cite{ChandlerWilde2013}, which contains similar pictures.  Indeed, one of their dimension $15$ pictures (their Figure 3, top right, p.~762) clearly shows rounding errors from nilpotent matrices, as explained for the skew-symmetric context in~\cite{Corless2021}. The authors of~\cite{ChandlerWilde2013} do not appear to have noticed that the central rose in that figure is due to rounding errors; to be sure, they only occur at odd dimension and in particular when the dimension is one less than a power of two, so they are easy to overlook (and are not central to the subject of their paper anyway).

The paper~\cite{ChandlerWilde2013} is described in Mathematical Reviews as ``a significant paper'', and now that we are aware of it, we agree. The paper studies infinite tridiagonal matrices with entries drawn from $\pm 1$ at random, and some of its results clearly transfer over to our case.  In particular, they also prove that the eigenvalues must lie in the diamond shape.  They do so by appealing to a known result about \textsl{numerical range}, proved in 1951 by Kippenhahn; a translation is available in~\cite{FZachlin2008}.  Very interestingly, these papers support our contention that the Bendixon--Bromwitch--Hirsch theorem should be better known, because the result of Kippenhahn that is used is exactly this theorem, which Kippenhahn rediscovered (apparently independently). Moreover, none of these papers appear to know that this result is from the first decade of the 20th century, not the middle.  To be fair, Kippenhahn then generalized the result substantially (after rediscovering it), and gave more tools than simple boxes to enclose eigenvalues. Indeed, the algebraic tools he invented are now called Kippenhahn polynomials, and we are very interested to see if they can be used in a Bohemian context.

The paper~\cite{Hagger2015} cites~\cite{ChandlerWilde2013} and proves that their conjecture about the asymptotic density of the \emph{general} problem they study is correct. This implies (among other things) that the spectra of all such matrices is dense on a \emph{disk}, as in the Tao and Vu result for general matrices.  It is only the peculiar special population cases studied here and in~\cite{ChandlerWilde2013} which converge instead to diamond-shaped subsets.

Finally, another of the referees asks if there are unsolved conjectures that arise from the pictures of this paper.  The answer is yes, certainly.  Perhaps the most ``itching'' question has to do with the visible spirals in figure~\ref{fig:unstableandstable}. We suspect that the spirals have something to do with the particular algebraic numbers that occur as eigenvalues of this family, along the lines of~\cite{Stange2017} or~\cite{harriss2020algebraic}.  But it's not so simple, and while other instances of spirals are known (see e.g. the January image of the Bohemian Matrix 2022 calendar), they are rare.

% We now leave you with more questions.  How many \textsl{unimodular} matrices are in any given family?  How many matrices have inverses that are also Bohemian (in the same family?) [In that case, we say that the matrix has \textsl{rhapsody}. Since number theory is the Queen of mathematics, this joke will never not be funny.] Can we solve the \textsl{inverse} eigenvalue problem over a given Bohemian family?  That is, given an eigenvalue, can we decide if there is a Bohemian matrix with that eigenvalue? This would help us to compute \emph{minimal height companion matrices}, for instance.  What questions would you like to ask, of a Bohemian family?

\begin{acks}
The Python code we used for many of the experiments was written largely by Eunice Y.~S.~Chan (with some help by Rob Corless) for another project; we thank her for her work on that code.  We also thank her for many discussions on Bohemian matrices and colouring algorithms.

We thank Mark Giesbrecht for hours of discussion on this topic.

We thank Neil J.~Calkin for valuable discussions about finding good questions to direct research.
We thank Ilias Kotsireas for pointing out several references relevant to the Hadamard matrix literature, and for helpful discussions.  We thank Nick Higham for an independent proof of Theorem~\ref{thm:bent}, and for searching out the original references to Hirsch, Bromwich, and Bendixon.  We thank Stefano Serra--Capizzano for useful comments on an earlier draft and for pointing out other references. We also thank Owen Maresh (@graveolens) for pointing out the connection to Kate Stange's work.

Finally, we thank the referees for their thorough and careful work during this time of climate crisis, war, and pandemic. We especially thank Referee 2 for pointing out the extremely relevant and interesting paper~\cite{ChandlerWilde2013}, which showed us a whole area of research that is relevant to Bohemian matrices.

Partially supported by NSERC grant RGPIN-2020-06438, and
partially supported by the grant PID2020-113192GB-I00 (Mathematical Visualization: Foundations, Algorithms and Applications) from the Spanish MICINN.
\end{acks}
% \bibliographystyle{plain}
% \bibliography{references}

\appendix

\section{Maple code}
The Maple code to implement our version of the algorithm of~\cite{Bottcher2021} can be found in the Maple workbook at \href{https://github.com/rcorless/Bohemian-Matrix-Geometry}{https://github.com/rcorless/Bohemian-Matrix-Geometry}.  The algorithm is laid out in Algorithm~\ref{alg:SS}.
\begin{figure}
    \centering
    \includegraphics[width=0.9\textwidth]{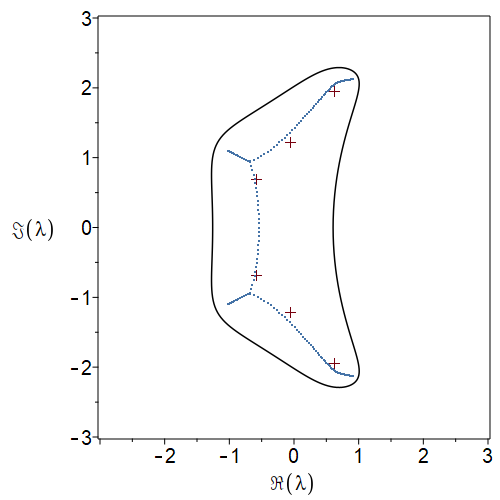}
    \caption{An example SS curve, together with the eigenvalues of the smallest matrix in the Toeplitz family which the Laurent polynomial symbol~\eqref{eq:LaurentPolynomial} applies. The closed curve is the graph of $a(e^{i\theta}/r)$ for $r=1.75$.}
    \label{fig:ToeplitzCup}
\end{figure}
As an example, we chose the vector $\mathbf{t} = [1, -1, 1, 0, 1]$, so $m=6$.  This means that the symbol is the Laurent polynomial
\begin{equation}
    a(z) = -\frac{1}{z} + \sum_{k=1}^{m-1} t_k z^k\>. \label{eq:LaurentPolynomial}
\end{equation}
% The pole at zero only has order $1$ because the matrix is not just Toeplitz but upper Hessenberg.  The residue is $-1$ because we have chosen this unit.  There is no constant term because that is just a shift of the eigenvalue.
For Figure~\ref{fig:ToeplitzCup} we chose $101$ values of $\phi$ equally-spaced on $[-\pi,\pi]$. We chose the scale factor $\rho = 1.75$ because it gives a reasonably tight bound on the SS curves when we draw the image $a(e^{i\psi}/r)$ of the unit circle under the map defined by the symbol.
We also plotted the eigenvalues of the dimension $m$ matrix with that population---this is the \emph{smallest} matrix in the family that this symbol and SS curve apply to.  The visible agreement of eigenvalues and SS curve is satisfactory.
\begin{algorithm}[H]
\caption{Schmidt--Spitzer, specialized to unit upper Hessenberg zero-diagonal case}
	\label{alg:SS}
	\begin{algorithmic}[1]
	\REQUIRE Scale factor $\rho > 1$. Vector $\mathbf{t}$ of Toeplitz matrix entries, of length $m-1$. $t_k \in \mathbb{C}$.
	\STATE Construct the scaled Laurent polynomial $a(z) = -\rho/z + \sum_{k=1}^{m-1} t_k (z/\rho)^k$
	\STATE Choose a vector of $\phi$ values in $-\pi \le \phi_j \le \pi $. More vector entries mean a finer resolution of the SS curves.
	We ignore the case $\phi=0$ which requires special handling but only adds isolated points.
	\FOR{$\phi_\ell \in \phi$}
	  \STATE Solve the polynomial equation $za(z) - za(e^{i\phi_\ell}z) = 0$. There are $m$ roots $u_j$.
	  \FOR{$j$ to $m$}
	  \STATE Compute $\lambda = a(u_j)$.
	  \STATE solve $za(z) - z\lambda$. There are $m$ roots $v_k$ again, two of which are $u_j$ and $e^{i\phi_\ell}u_j$.
	  \STATE If $|u_j|=|e^{i\phi_\ell}u_j|$ are the \emph{smallest} roots in magnitude, then $\lambda$ is on a SS curve. Record it and continue.
	  \ENDFOR
	\ENDFOR
	\end{algorithmic}
\end{algorithm}

% The code base used here is a combination of Maple (including the Image and ColorTools packages) and Python.  They are interoperable to an extent, in that data built in one can be saved to a JSON container and read into the other.  Figure~\ref{fig:symmetric6} and Figures~\ref{fig:stable}--\ref{fig:unstableandstable} were produced using polynomials computed in Python and transferred to Maple for image processing. All others were produced in Python. For some things the Python code is faster, but for others (including polynomial manipulation) Maple is faster.  An early version of the Maple code is available at~\cite{Corless:2021:WhatCan}.  The Python code will be released soon.  Older software for Bohemian matrices can be found at \url{https://github.com/BohemianMatrices}.

% We have not yet systematically taken advantage of parallelism, even though many of the computations here are trivially parallel. One example (not shown) of computing the eigenvalues of $2^{31}$ matrices on a 32 core machine was successful; and the enhanced plots in Figures~\ref{fig:stable} and~\ref{fig:Rayleigh} used GPU computing on Google's COLAB with TensorFlow.  It is clear that there is significant scope for parallel computation.

\end{document}